\documentclass[
journal]{IEEEtran}                                      
\usepackage{rotating}
\usepackage{amsmath}
\ifCLASSINFOpdf
\usepackage{graphicx,amssymb}
\usepackage{algorithm}
 \usepackage{algorithmic}

\DeclareGraphicsExtensions{.pdf,.jpg,.jpeg,.png}
\else
\usepackage[dvips]{graphicx}
\usepackage{epsfig}
\DeclareGraphicsExtensions{.eps}
\fi

\usepackage{pdfsync}
\usepackage{amsthm}

\newtheorem{theorem}{Theorem}[section]
\newtheorem{lemma}[theorem]{Lemma}

\newtheorem{corollary}[theorem]{Corollary}

\newtheorem{remark}[theorem]{Remark}

\def\CC{ {\mathbb C} }
\def\RR{ {\mathbb R} }
\def\ZZ{ {\mathbb Z} }
\numberwithin{equation}{section}

\begin{document}


%
\title{Phase Retrieval from  Multiple-Window  Short-Time Fourier Measurements}
\author{Lan Li, Cheng Cheng,  Deguang Han, Qiyu Sun (Member, IEEE) and Guangming Shi (Senior Member, IEEE)
\thanks{Li and Shi is with the School of Electronic Engineering, Xidian Univeristy, Xi'an 710071, China;
\and Cheng,  Han and Sun are with the  Department of Mathematics,
University of Central Florida,
Orlando 32816, Florida, USA.
Emails:  lanli98@126.com; cheng.cheng@knights.ucf.edu;  deguang.han@ucf.edu; qiyu.sun@ucf.edu; gmshi@xidian.edu.cn. This work is partially supported by the National Natural Science Foundation of China (Grant No. 61472301), the Research Fund for the Doctoral Program of Higher Education (No. 20130203130001), and the  National Science Foundation (DMS-1403400 and DMS-1412413).}}

\maketitle

\begin{abstract}
In this paper, 
we introduce two symmetric directed
graphs depending on supports of  signals and  windows, and we show that the  connectivity of those graphs provides either necessary or sufficient conditions
 to  phase retrieval of a signal  from magnitude measurements  of its multiple-window short-time Fourier transform. 
Also we propose an algebraic reconstruction algorithm, and  provide  an error estimate to our algorithm
when  magnitude measurements are corrupted  by deterministic/random noises.
   \end{abstract}


\vskip-1.5mm

\section{\bf Introduction}  

Phase retrieval  considers recovering a signal of interest from magnitudes of its (non)linear measurements.
It arises in various fields of science and engineering, such as  X-ray crystallography, coherent diffractive imaging, optics
 and more. 
 The underlying recovery  is  an ill-posed problem inherently. The signal could be  reconstructed,
 in an  efficient and robust manner,
 only  if  we  have  additional  information
about the signal  (\cite{walther63}--\cite{jaganathany15}).
  In this paper, we discuss the  phase retrieval problem for   $N$-dimensional complex signals 
\vskip-0.1in
\begin{equation}
{\bf x}=(x(0), x(1), 
\ldots, x(N-1))^\top \in \mathbb{C}^N\end{equation}
with some constraints on their supports,
\vskip-0.1in
\begin{equation} \label{graph1.def1}
V({\bf x})=\big\{n: 
\ x(n)\ne 0\big\}.\end{equation}

Given a nonzero window ${\bf w}=(w(0),w(1),\ldots,w(N-1))^\top$
  with period $N$ extension and a separation parameter $L$  between  adjacent  short-time  sections,
the short-time Fourier transform (STFT) of a  signal ${\bf x}$  is given by
\vskip-0.1in
\begin{equation}\label{stft.def} X_{{\bf w}}(Lm,k)=\frac{1}{N}\sum_{n=0}^{N-1}x(n)w(Lm-n)e^{-i2\pi kn/N},
\end{equation}
where $0\le m\le N/L-1$ and  $0\le k\le N-1$.  The
STFT has been widely used in signal/imaging processing (\cite{cohenbook, grochenigbook}).
In this paper, inspired by applications in microscopy and optical imaging, we consider reconstructing the signal ${\bf x}$ from
magnitude measurements of its multiple-window  STFT,
\begin{equation}\label{stft.data}
|X_{{\bf w}_r}(Lm,k)|, 1\le r\le R, 0\le m\le {N}/{L}-1, 0\le k\le N-1,
\end{equation}
\vskip-0.2in\noindent
where  ${\bf W}=\{ {\bf w}_r\}_{r=1}^{R}$ is a family of windows with period $N$ extension.
The above reconstruction problem has been explored with various approaches
 (\cite{lim79}--\cite{eldararxiv2015}). 
The  special case with $L=N$
is also known as phase retrieval with structured  illuminations  and  masks
(\cite{loewenbook}--\cite{candes14}).

Define the supporting length of a nonzero window  ${\bf w} $ 
with period $N$ extension by  
\vskip-0.1in
\begin{eqnarray}\label{windowsupp.def}
\label{supporting.def}  \hskip-0.18in  l({\bf w}) & \hskip-0.08in = & \hskip-0.08in \min_{0\le  l\le N-1}\big\{l+1: \ \ {\rm there \ exists}\ n'\ {\rm so\ that}\  \nonumber \\
  & \hskip-0.08in  &  \  w(n)=0 \ {\rm for\ all} \ n\not\in [n', n'+l]+N\ZZ\big\}.  \end{eqnarray}
 It is observed in \cite[Theorem 2]{eldarIEEESPL2015}
that not all $N$-dimensional signals
can be recovered, up to a global phase,  from
magnitude measurements  \eqref{stft.data}
of their multiple-window STFT if all windows ${\bf w}_r, 1\le r\le R$,
have supporting length  less than $N/2$, cf.  \cite{ccsw2016} for similar phenomenon observed when recovering signals in a shift-invariant space
from  magnitudes of their sampling data.
 In Section \ref{necessary.section} of this paper, we introduce a symmetric directed graph ${\mathcal G}({\bf x}, {\bf W}, L)$ 
 with  $V({\bf x})$ in \eqref{graph1.def1} as its vertex set,
and we show  in Theorem \ref{necessary.thm} that connectivity of the  above graph  
is a necessary condition to  phase retrievability of the signal ${\bf x}$ from
magnitude measurements  \eqref{stft.data}
of its multiple-window STFT.

 A fundamental question in  
  phase retrieval  is whether a signal  
  is uniquely determined, up to a global phase,
  by its noiseless 
   measurements \eqref{stft.data}.  For $L=R=1$, a sufficient condition was proposed in \cite[Theorem 1]{eldarIEEESPL2015}
  to recover a signal ${\bf x}$ with all  components being nonzero   from magnitude measurements  \eqref{stft.data}
of its STFT, cf. \cite[Theorem 2.4]{bojarovska2016}.
In Section \ref{sufficient.section} of this paper, we introduce a  symmetric directed subgraph $\tilde {\mathcal G}({\bf x}, {\bf W}, L)$ in \eqref{graph2.def},
and we prove in Theorem \ref{sufficient.mainthm} that, under mild  conditions
on the window family ${\bf W}$,
 connectivity of the graph  $\tilde {\mathcal G}({\bf x}, {\bf W}, L)$   is a sufficient condition to
reconstruct  the  signal ${\bf x}$, up to a global phase, from
magnitude measurements  \eqref{stft.data}
of its multiple-window STFT.  Applying Theorem \ref{sufficient.mainthm} with $V({\bf x})=\{0, 1, \ldots, N-1\}$ and $L=R=1$  leads to the result in  \cite[Theorem 1]{eldarIEEESPL2015}, see Corollary \ref{sufficient.mainthm.fullL=1}.

Consider the scenario that  magnitude measurements \eqref{stft.data}
 of the multiple-window STFT are corrupted by deterministic/random noises
 ${\pmb \epsilon}=({\pmb \epsilon}(r, m,k))$ with level $|{\pmb \epsilon}|$,
 \begin{equation}\label{noise.def}
 Y(r, m, k):= |X_{{\bf w}_r}(Lm,k)|^2+\ {\pmb \epsilon}(r, m,k),
 \end{equation}
where
$|{\pmb \epsilon}| = \max\{ |{\pmb \epsilon}(r, m,k)|: \  1\le r\le R,  0\le m\le N/L-1, 0\le k\le N-1\}$.
 Another fundamental issue in phase retrieval is
to design efficient and robust algorithms so that   a good approximation ${\bf x}_{\pmb \epsilon}$ to the original signal ${\bf x}$, up to a global phase,
 could be found
 when only noisy  
 measurements \eqref{noise.def} are available.
Designing such  reconstruction algorithms is a great challenge  in general, and
 several algorithms have been proposed in the literature, see \cite{Fienup82, eldararxiv2015, gerchbert72, berdory15} and references therein.
In Section  \ref{algorithm.section} of this paper, we propose an algebraic reconstruction algorithm
from noisy  measurements \eqref{noise.def}, and we establish an error estimate in Theorem \ref{arror.thm}
when the graph  $\tilde {\mathcal G}({\bf x}, {\bf W}, L)$  in \eqref{graph2.def} is connected.

Notation:  ${\bf A}^H$ is  Hermitian of a matrix ${\bf A}$; ${\bf a}\circ {\bf b}$ is
the componentwise (Hadamard) product of vectors ${\bf a}$ and ${\bf b}$;
$\lfloor t\rfloor$ is  the largest integer less than or equal to $t$, and  $k\ {\rm mod}\ N$ is  the remainder of the Euclidean division of  an integer $k$ by $N$.


\section {A Necessary Condition on Phase Retrieval}
\label{necessary.section}

\vskip-0.03in

Given an $N$-dimensional complex signal ${\bf x}$,
a family  ${\bf W}=\{ {\bf w}_r\}_{r=1}^{R}$  of window functions with period $N$ extension,
and  a separation parameter $L$  with $N/L\in \ZZ$,
we define a  graph
\begin{equation}\label{graph1.def}
{\mathcal G}({\bf x}, {\bf W}, L):=\big(V({\bf x}), E({\bf W}, L)\big)\end{equation}
\vskip-0.02in
\noindent with
\begin{eqnarray} \label{graph1.def2}
\hskip-0.28in  & \hskip-0.08in &  \hskip-0.08in E({\bf W}, L) :=      \Big\{(n, n')\in  V({\bf x})\times V({\bf x}): \ n\ne n' \ {\rm and}\nonumber\\
\hskip-0.28in & \hskip-0.08in & \quad  \sum_{r=1}^{R}\sum_{m=0}^{N/L-1} |w_r(Lm-n')w_r(Lm-n)|^2\ne 0\Big\},\end{eqnarray}
\vskip-0.02in
\noindent where 
 $V({\bf x})$  is given in \eqref{graph1.def1}  and
${\bf w}_r=((w_r(0), \ldots, w_r(N-1))^\top, 1\le r\le R$.  The  symmetric directed graph
${\mathcal G}({\bf x}, {\bf W}, L)$ has indices of nonzero components of the signal ${\bf x}$ as its vertices,
and  it has edges between  two distinct vertices $n$ and $n'$  only if  
$w_r(Lm-n')w_r(Lm-n)\ne 0$
for some    $1\le r\le R$ and $0\le m\le N/L-1$.


 \begin{theorem} \label{necessary.thm}
 Let   ${\bf W}=\{{\bf w}_r\}_{r=1}^R$   be a family of window functions with period $N$ extension,
 and
$L\ge 1$ be a separation parameter with $N/L\in \ZZ$.  If ${\bf x}\in \CC^N$ can be determined, up to a global phase,  from
magnitude measurements \eqref{stft.data}
of its multiple-window STFT, then
the graph ${\mathcal G}({\bf x}, {\bf W}, L)$ in \eqref{graph1.def} is connected.
 \end{theorem}

 \begin{proof} Suppose, on the contrary, that
 ${\mathcal G}({\bf x}, {\bf W}, L)$ in \eqref{graph1.def} is disconnected.
Then there exists a subset $V_1\subset V({\bf x})$  such that  $V_1\ne \emptyset, V({\bf x})\backslash V_1\ne \emptyset$, and there are no edges between vertices in $V_1$ and $V({\bf x})\backslash V_1$.
Let ${\bf x}_{V_1}\in \CC^N$ be the signal which  coincides with ${\bf x}$ on the indices in $V_1$ and is extended to zeros in $\{0, 1, \ldots, N-1\}\backslash V_1$. 
 Observe from \eqref{graph1.def}
that for any $1\le r\le R$ and $0\le m\le N/L-1$,
  there is an edge between  two indices of  nonzero components of ${\bf x}_\theta \circ {\bf w}_{r, Lm}$,
 where
 $${\bf w}_{r, Lm}= (w_r(Lm), \ldots, w_r(Lm-N+1))^\top.$$
Therefore
either
${\bf x}_\theta \circ {\bf w}_{r, Lm}=  e^{-2\pi i \theta} {\bf x}_{V_1} \circ {\bf w}_{r, Lm}
$
or
${\bf x}_\theta \circ {\bf w}_{r, Lm}=   ({\bf x}- {\bf x}_{V_1}) \circ {\bf w}_{r, Lm}$ by the construction of  $V_1$, where
$${\bf x}_\theta= e^{-2 \pi i\theta} {\bf x}_{V_1}+({\bf x}-{\bf x}_{V_1}), \ \theta\in \RR.$$
  Hence magnitude measurements of the multiple-window STFT  
 of the signal ${\bf x}_\theta$
are independent on $\theta\in \RR$.
This, together with  ${\bf x}_0={\bf x}$ and the  phase retrievability assumption, implies that
${\bf x}_{1/2}
= e^{-2\pi i\beta} {\bf x}$  for some $\beta\in \RR$.
Thus
\begin{equation}\label{signalv1.eq} (1+e^{-2\pi i\beta}) {\bf x}_{V_1}= (1-e^{-2\pi i\beta}) ({\bf x}-{\bf x}_{V_1}). \end{equation}
For the case that $e^{-4\pi i\beta}=1$, 
either ${\bf x}_{V_1}$ or ${\bf x}-{\bf x}_{V_1}$ is a zero signal, 
 which is a contradiction.  For the remaining case that
$e^{-4\pi i\beta}\ne 1$,
it follows from \eqref{signalv1.eq} that ${\bf x}_{V_1}$ and ${\bf x}-{\bf x}_{V_1}$ have the same support, which contradicts to the construction of ${\bf x}_{V_1}$.
\end{proof}

Given a window  family ${\bf W}$,  
 the graph ${\mathcal G}({\bf x}, {\bf W}, L)$ in \eqref{graph1.def} could be disconnected
for some signals ${\bf x}$.
As an application of Theorem \ref{necessary.thm}, we have the following result on phase retrievability,
cf.  \cite[Theorem 2]{eldarIEEESPL2015}, \cite[Proposition 2.3]{bojarovska2016}  and 
  \cite 
{nawabieee83}.


\begin{corollary}\label{necessary.cor}
 Let   ${\bf W}=\{ {\bf w}_r\}_{r=1}^{R}$   be a family of window functions with period $N$ extension such that
 $l({\bf w}_r)\le N/2$ for all $1\le r\le R$.
 Then not all $N$-dimensional signals 
can be determined, up to a global phase,  from
magnitude measurements \eqref{stft.data}
of their multiple-window STFT.
\end{corollary}

\begin{proof}
Let  
${\bf x}_0$ 
be the signal having
 $0$-th and $\lfloor N/2\rfloor$-th components as one and other components  zero. Then the corresponding graph
${\mathcal G}({\bf x}_0, {\bf W}, L)$ in \eqref{graph1.def} has two vertices $0$ and $\lfloor N/2\rfloor$.
From the supporting length assumption 
for  
${\bf w}_r, 1\le r\le R$,
it follows that
$w_r(n)w_r(n-\lfloor N/2\rfloor)=0$ for all $0\le n\le N-1$.
Hence 
${\mathcal G}({\bf x}_0, {\bf W}, L)$ is disconnected. This together with  Theorem \ref{necessary.thm} completes the proof.
\end{proof}

\vspace{-2mm}

\section {A Sufficient Condition for Phase retrieval}
\label{sufficient.section}

Given an $N$-dimensional complex signal ${\bf x}$,
a family  ${\bf W}=\{ {\bf w}_r\}_{r=1}^{R}$  of window functions with period $N$ extension,
and a separation parameter $L$ with $N/L\in \ZZ$,
we define a  graph
\begin{equation}\label{graph2.def}
\tilde {\mathcal G}({\bf x}, {\bf W}, L):=\big(V({\bf x}), \tilde E({\bf W}, L)\big)\end{equation}
with
\begin{eqnarray} \label{graph2.def2} \hskip-0.18in  &\hskip-0.08in & \hskip-0.08in \tilde E({\bf W}, L):=\Big\{(n, n')\in V({\bf x})\times V({\bf x}):\  n\ne n'
\ {\rm and} \  \nonumber\\
\hskip-0.18in  & \hskip-0.08in &\quad  n, n'\in Lm-a({\bf w}_r) -\{0, l({\bf w}_r)-1\}+N\ZZ \  {\rm for }\   \nonumber
\\
\hskip-0.18in  & \hskip-0.08in & \qquad\   {\rm some}\  0\le m\le N/L-1 \ {\rm  and} \  1\le r\le R\Big\},\end{eqnarray}
where  $V({\bf x})$ is given in \eqref{graph1.def1} and supporting intervals $[a({\bf w}_r), a({\bf w}_r)+l({\bf w}_r)-1]+N\ZZ$ of  windows ${\bf w}_r, 1\le r\le R$,
are so chosen  that
 $0\le a({\bf w}_r)\le N-1$,
 \vspace{-1mm}
  \begin{equation}\label{supporting.def2}
  w_r( a({\bf w}_r)) \overline {w_r(a({\bf w}_r)+l({\bf w}_r)-1)}\ne 0,
  \vspace{-3mm}
  \end{equation}
  and
  \vspace{-2mm}
  \begin{equation} \label{supporting.def3}
  w_r(n)=0 \ {\rm for\ all} \  n\not\in [a({\bf w}_r), a({\bf w}_r)+l({\bf w}_r)-1]+N\ZZ.
  \end{equation}
  The existence and uniqueness of  
  $a({\bf w}_r)$
  follow  
  from \eqref{windowsupp.def}. By \eqref{graph1.def2} and \eqref{supporting.def2}, we see that 
$\tilde {\mathcal G}({\bf x}, {\bf W}, L)$ 
 is a symmetric directed subgraph of the graph
${\mathcal G}({\bf x}, {\bf W}, L)$ in \eqref{graph1.def}.

\vspace{-1mm}
 \begin{theorem}\label{sufficient.mainthm}
 Let  ${\bf x}\in {\CC}^N $, $L$ be a separation parameter with $N/L\in \ZZ$,
  and  let ${\bf W}=\{ {\bf w}_r\}_{r=1}^{R}$   be a family of window functions with period $N$ extension
 such that 
 \begin{equation} \label{sufficient.mainthm.eq1}
 l({\bf w}_r)\le N/2
 \end{equation}
 for all $ 1\le r\le R$, and
 \begin{equation}\label{sufficient.mainthm.eq2}
 {\bf A}_{m}= \Big(\beta_r\big(m+ j{N}/{L}\big)\Big)_{1\le r\le R, 0\le j\le L-1} \end{equation}
 have rank $L$ for all $ 0\le m\le N/L-1$, where
 \begin{equation}\label{sufficient.mainthm.eq3} \beta_r(k)=\frac{1}{N}\sum_{n=0}^{N-1} |w_r(n)|^2 e^{-2\pi i kn/N}, \ 0\le k\le N-1.\end{equation}
 If
 the graph $\tilde {\mathcal G}({\bf x}, {\bf W}, L)$ in \eqref{graph2.def} is connected, then
  ${\bf x}$ can be recovered, up to a global phase,  from
magnitude measurements \eqref{stft.data}
of its multiple-window STFT.
 \end{theorem}

For $L=1$, the full rank requirement \eqref{sufficient.mainthm.eq2}  
 becomes
\begin{equation}\label{L=1.eq}
\sum_{r=1}^R |\beta_r(m)|^2\ne 0\ {\rm  for\ all} \  0\le m\le N-1,\end{equation}
and   the set $\tilde E({\bf W}, L)$ of edges 
can be rewritten as
\begin{eqnarray} \label{graph2.def2.L=1} \tilde E({\bf W}, 1) &\hskip-0.08in := & \hskip-0.08in \Big\{(n, n'):\ n- n'=\pm ( l({\bf w}_r)-1)+N\ZZ
\nonumber\\
& & \quad {\rm  for\ some} \  1\le r\le R\Big\}.\end{eqnarray}
%
If we further assume that the signal ${\bf x}$ has its all components being nonzero (i.e.,
$V({\bf x})=\{0, \ldots, N-1\}$), one may verify from \eqref{graph2.def2.L=1} that
 $\tilde {\mathcal G}({\bf x}, {\bf W}, 1)$  
is connected if and only if
\begin{equation}\label{prime.eq} l({\bf w}_1)-1, \ldots, l({\bf w}_R)-1\ {\rm and} \ N \ {\rm are\ coprime}.
\end{equation}
Therefore applying Theorem \ref{sufficient.mainthm} with $L=1$, 
we obtain the following result,
which is given in \cite[Theorem 1]{eldarIEEESPL2015} for $R=1$, cf. \cite[Theorems 2.4]{bojarovska2016}.

\begin{corollary}\label{sufficient.mainthm.fullL=1}
 Let  ${\bf x}\in {\CC}^N$  have its all components being nonzero,
  and  ${\bf W}$ 
   be a family of window functions having period $N$ extension
  and satisfying
\eqref{sufficient.mainthm.eq1}, \eqref{L=1.eq} and \eqref{prime.eq}.
Then  ${\bf x}$ can be recovered, up to a global phase,  from
magnitude measurements \eqref{stft.data}
of its multiple-window STFT.
 \end{corollary}

For  $L=N$, the full rank requirement \eqref{sufficient.mainthm.eq2}
can be rewritten as
\begin{equation}\label{sufficient.mainthm.eq2L=N}
{\rm the\  matrix} \ \big(|w_r(n)|^2\big)_{1\le r\le R, 0\le n\le N-1} \ {\rm has \ rank} \ N.
\end{equation}
Then applying Theorem \ref{sufficient.mainthm} with $L=N$ yields the following result on   phase retrievability with structured illuminations
and masks, cf. \cite{loewenbook}--\cite{candes14}.

\begin{corollary}\label{sufficient.mainthm.fullL=N}
 Let ${\bf W}$ 
   be a family of window functions having period $N$ extension
  and satisfying
\eqref{sufficient.mainthm.eq1} and \eqref{sufficient.mainthm.eq2L=N}.
Then any signal
 ${\bf x}\in {\CC}^N$ with $\tilde {\mathcal G}({\bf x}, {\bf W}, N)$ being connected
can be recovered, up to a global phase,  from
magnitude measurements \eqref{stft.data}
of its multiple-window STFT.
 \end{corollary}

 To prove Theorem \ref{sufficient.mainthm}, we need a technical lemma.  

 \begin{lemma}\label{sufficient.lem1}
 Let $L$ be a separation parameter with $N/L\in \ZZ$ and  ${\bf W}=\{ {\bf w}_r\}_{r=1}^{R}$   be a family of window functions having period $N$ extension
 and satisfying \eqref{sufficient.mainthm.eq2}.
 Then
magnitudes of any  signal ${\bf x}=(x(0),x(1),x(2),\ldots,x(N-1))^\top\in {\CC}^N $ can be recovered from
magnitude measurements \eqref{stft.data}
of its  multiple-window STFT. Moreover,
 \begin{eqnarray}  \label{sufficient.thm1.eq2}
& \hskip-0.08in &  \hskip-0.08in  |x(n)|^2 = \frac{L}{N} \sum_{m, m'=0}^{N/L-1}  \sum_{j, j'=0}^{L-1} e^{-2\pi i (m (m' L-n)/N-jn/L)}\nonumber\\
 & \hskip-0.08in &  \hskip-0.08in \quad  \times  a_m(j,j')
\Big(\sum_{r=1}^R \overline{\beta_r(m+j'N/L)}   Z({\bf w}_r, m')\Big)
 \end{eqnarray}
 for all $0\le n\le N-1$,
 where  $$({\bf A}_m^H {\bf A}_m)^{-1}=(a_m(j, j'))_{0\le j, j'\le L-1}$$
  and
 $$Z({\bf w}_r, m)= \sum_{k=0}^{N-1}|X_{{\bf w}_r}(Lm,k)|^2 ,\  0\le m\le N/L-1. $$
 \end{lemma}


%

%

We postpone the proof of Lemma  \ref{sufficient.lem1} to the end of this section and start the proof of Theorem \ref{sufficient.mainthm}.

\begin{proof}[Proof of Theorem \ref{sufficient.mainthm}]
By Lemma \ref{sufficient.lem1},
 $|x(n)|^2, 0\le n\le N-1$, are determined from 
$|X_{{\bf w}_r}(Lm,k)|^2, 1\le r\le R, 0\le m\le L/N-1, 0\le k\le N-1$.
Therefore it remains to find $x(n)/|x(n)|, n\in  V({\bf x})$, up to a global phase.
From connectivity of the 
 graph $\tilde {\mathcal G}({\bf x}, {\bf W}, L)$,
it suffices to show that for endpoints $n_1, n_2$ of any edge,
the phase difference between  $x(n_1)/|x(n_1)|$ and $x(n_2)/|x(n_2)|$ is determined from
magnitude measurements \eqref{stft.data}
of the multiple-window STFT. 

By the assumption on vertices $n_1$ and $n_2$, there exist $1\le r\le R$ and $0\le m\le N/L-1$ such that
$l({\bf w}_r)\ge 2 $ 
and
\begin{equation}\label{n1n2.eq} n_1, n_2\in Lm-a({\bf w}_r)  -\{0, l({\bf w}_r)-1\}+N\ZZ.\end{equation}
Without loss of generality, we assume that
\begin{equation}  \label{sufficient.mainthm.pf.eq2}
n_1\in Lm- a({\bf w}_r)+N\ZZ\ {\rm and} \ n_2\in Lm- a({\bf w}_r)-l({\bf w}_r)+1+N\ZZ.
\end{equation}
By \eqref{supporting.def2}, \eqref{supporting.def3} and 
\eqref{sufficient.mainthm.eq1}, we have
\begin{equation}  \label{sufficient.mainthm.pf.eq3}
w_r(n) \overline{w_r(n+l({\bf w}_r)-1)}
\ne 0\ {\rm if \  and \ only\ if} \ n\in a({\bf w}_r)+N\ZZ.
\end{equation}
From \eqref{sufficient.mainthm.pf.eq3} we obtain
\begin{eqnarray} \label{sufficient.mainthm.pf.eq4-}
 &\hskip-0.08in  &  \hskip-0.08in N\sum_{k=0}^{N-1}|X_{{\bf w}_r}(Lm,k)|^2 e^{i2\pi k (l({\bf w}_r)-1)/N}\nonumber\\
& \hskip-0.08in=& \hskip-0.08in  \sum_{n=0}^{N-1} x(n + l({\bf w}_r)-1 \ {\rm mod }\ N) \overline{x(n)}\nonumber\\
& \hskip-0.08in & \hskip-0.08in \quad  \times w_r(Lm-n- l({\bf w}_r)+1) \overline{w_r(Lm-n)}\nonumber\\
& \hskip-0.08in = & \hskip-0.08in x(n_1)  \overline{x(n_2)} w_r(a({\bf w}_r)) \overline{w_r(a({\bf w}_r)+l({\bf w}_r)-1)}.\nonumber\\
\end{eqnarray}
Therefore
\begin{eqnarray} \label{sufficient.mainthm.pf.eq4}
\hskip-0.2in & \hskip-0.2in & \hskip-0.1in\frac{x(n_1)}{|x(n_1)|} \cdot  \frac{\overline{x(n_2)}}{|x(n_2)|}= \frac{w_r(a({\bf w}_r)+l({\bf w}_r)-1)\overline{w_r(a({\bf w}_r)) }} {|w_r(a({\bf w}_r)+l({\bf w}_r)-1) \overline{w_r(a({\bf w}_r))} |}  
\nonumber\\
\hskip-0.2in&  \hskip-0.2in & \quad  \qquad \times
\frac{\sum_{k=0}^{N-1}|X_{w_r}(Lm,k)|^2 e^{i2\pi k (l({\bf w}_r)-1)/N}}
{\Big|\sum_{k=0}^{N-1}|X_{w_r}(Lm,k)|^2 e^{i2\pi k (l({\bf w}_r)-1)/N}\Big|}.
\end{eqnarray}
Combining \eqref{sufficient.mainthm.pf.eq2}, \eqref{sufficient.mainthm.pf.eq3} and \eqref{sufficient.mainthm.pf.eq4}
shows that the phase difference between $x(n_1)/|x(n_1)|$ and $x(n_2)/|x(n_2)|$
 is determined from
magnitude measurements \eqref{stft.data}
of the multiple-window STFT.
This 
completes the proof.
\end{proof}

We finish this section with the proof of Lemma \ref{sufficient.lem1}.

\begin{proof} [Proof of Lemma \ref{sufficient.lem1}]  For  $0\le m'\le N/L-1$ and $1\le r\le R$, we have
\begin{eqnarray*}
Z({\bf w}_r, m')  
&\hskip-0.08in =& \hskip-0.08in
  \frac{1}{N} \sum_{n=0}^{N-1}|x(n)|^2|w_r(Lm'-n)|^2\\
& \hskip-0.08in  = & \hskip-0.08in   \sum_{k=0}^{N-1}   \alpha(k) \beta_r(k) e^{2\pi i m'kL/N},
\end{eqnarray*}
where
$$\alpha(k)=\frac{1}{N}\sum_{n=0}^{N-1} |x(n)|^2 e^{-2\pi i kn/N}, \ 0\le k\le N-1.$$
Hence  for  $0\le m\le N/L-1$ and $1\le r\le R$, we obtain
\begin{eqnarray}\label{sufficient.thm1.pf.eq2}
\hskip-0.1in & \hskip-0.1in & \hskip-0.1in \frac{L}{N} \sum_{m'=0}^{N/L-1} Z({\bf w}_r, m') e^{-2\pi i m m' L/N}\nonumber\\
\hskip-0.1in & \hskip-0.1in &  \hskip-0.1in\qquad  =   \sum_{j'=0}^{L-1} \alpha(m+j'N/L) \beta_r(m+j'N/L).
\end{eqnarray}
Then  we get
\begin{eqnarray*}
\hskip-0.22in \alpha\big(m+jN/L\big) & \hskip-0.08in = & \hskip-0.08in \frac{L}{N}
\sum_{j'=0}^{L-1}\sum_{m'=0}^{N/L-1}  a_m(j,j')  e^{-2\pi i m m' L/N}\nonumber\\
& \hskip-0.08in &  \hskip-0.08in  \times
\Big(\sum_{r=1}^R \overline{\beta_r(m+j'N/L)}   Z({\bf w}_r, m')\Big),
\end{eqnarray*}
where $0\le m\le N/L-1$ and $0\le j\le L-1$.
This together with
\begin{equation} \label{sufficient.thm1.pf.eq1}
|x(n)|^2=\sum_{k=0}^{N-1} \alpha(k)  e^{2\pi i kn/N}, \ 0\le n\le N-1,
\end{equation}
proves \eqref{sufficient.thm1.eq2}.
\end{proof}

\vskip-0.08in

\section {Reconstruction algorithm and error estimates}
\label{algorithm.section}

Consider
the family ${\bf W}=\{{\bf w}_r\}_{r=1}^R$  of window functions  having
 period $N$ extension and satisfying \eqref{sufficient.mainthm.eq1} and \eqref{sufficient.mainthm.eq2}.
From Theorem   \ref{sufficient.mainthm}, it follows that
 any signal ${\bf x}=(x(0), \ldots, x(N-1))^\top$ with a connected graph $\tilde {\mathcal G}({\bf x}, {\bf W}, L)$
 can be reconstructed, up to a global phase, from
 magnitude measurements  $|X_{{\bf w}_r}(Lm,k)|, 1\le r\le R, 0\le m\le N/L-1, 0\le k\le N-1$,
 of its multiple-window STFT.
From the constructive proof of Theorem \ref{sufficient.mainthm},  we propose the following reconstruction algorithm:

 \begin{itemize}
 \item[{1.}]   Apply \eqref{sufficient.thm1.eq2} to find magnitudes $|x(n)|, 0\le n\le N-1$.

 \item[{2.}]  Create the graph $\tilde G({\bf x}, {\bf W}, L)$ in \eqref{graph2.def} and verify its connectivity.

 \item [{3.}] Apply  \eqref{sufficient.mainthm.pf.eq4} to find  phase difference between
 $\frac{x(n_1)}{|x(n_1)|}$ and $\frac{x(n_2)}{|x(n_2)|}$, where $n_1, n_2$ are endpoints of an
 edge of the graph $\tilde G({\bf x}, {\bf W}, L)$, provided that $\tilde G({\bf x}, {\bf W}, L)$ is connected.

 \end{itemize}

The  reconstruction algorithm proposed above indicates that 
 ${\bf x}$ can be recovered, up to a global phase, from its
$2NR/L$ measurements
 $ \sum_{k=0}^{N-1}|X_{w_r}(Lm,k)|^2 e^{i2\pi k (l({\bf w}_r)-1)/N}$ and $\sum_{k=0}^{N-1}|X_{w_r}(Lm,k)|^2$, where
$1\le r\le R, 0\le m\le N/L-1$.

For a window family ${\bf W}=\{{\bf w}_r\}_{r=1}^R$ with period $N$ extension,
 we set
  $\|{\bf W}\|_2=\big(\sum_{r=1}^R\sum_{n=0}^{N-1} |w_r(n)|^2\big)^{1/2} $,
     $\|{\bf W}\|_*=\min_{1\le r\le R} |w_r(a({\bf w}_r)) w_r(a({\bf w}_r)+l({\bf w}_r)-1)|$, and
     $\|{\bf A}\|_1=\sum_{m=0}^{N/L-1}\sum_{j,j'=0}^{L-1}|a_m(j,j')|$.
  For  the scenario that magnitude measurements \eqref{stft.data} of the multiple-window STFT are corrupted, 
  we have the following error estimate between the original signal  ${\bf x}$ and
the  approximation  ${\bf x}_{\pmb \epsilon}$ obtained from the proposed reconstructed algorithm with the corrupted  magnitude  measurements  \eqref{noise.def}.

\begin{theorem}\label{arror.thm}
 Let   $L, {\bf W}$ and ${\bf x}=(x(0), \ldots, x(N-1))^\top$  be as in Theorem \ref{sufficient.mainthm},
and
${\bf x}_{\pmb \epsilon}=(x_{\pmb \epsilon}(0), \ldots, x_{\pmb \epsilon}(N-1))^\top$  be the  approximation obtained from the proposed reconstructed algorithm with the corrupted  magnitude data  \eqref{noise.def}.  
           If
\begin{eqnarray}\label{error.thm.eq1}
\hskip-0.18in |{\pmb \epsilon}| & \hskip-0.08in \le  & \hskip-0.08in \frac{\min_{n\in V({\bf x})} |x(n)|^2}
{4 \|{\bf A}\|_1 \|{\bf W}\|_2^2},
\end{eqnarray}
 then there exists $\beta\in \RR$ such that
 \begin{equation}\label{error.thm.eq2}
 \big||x_{\pmb \epsilon}(n)|^2- |x(n)|^2\big| \le   \|{\bf A}\|_1 \|{\bf W}\|_2^2
  |{\pmb \epsilon}| 
 \end{equation}
  for all $0\le n\le N-1$,
 and
 \begin{equation}\label{error.thm.eq3}
 \Big|\frac{x_{\pmb \epsilon}(n)}{|x_{\pmb \epsilon}(n)|} - e^{2\pi i\beta}\frac{x(n)}{|x(n)|} \Big|
 \le \frac{2N^3 |{\pmb \epsilon}|}{\|{\bf W}\|_* \min_{n\in V({\bf x})} |x(n)|^2}
  \end{equation}
  for all $n\in V({\bf x})$.
\end{theorem}

\begin{proof}
The estimate \eqref{error.thm.eq2} follows   
from \eqref{sufficient.thm1.eq2}
and
$$ \sum_{r=1}^R |\beta_r(k)|\le \frac{\|{\bf W}\|_2^2}{N}, \ 0\le k\le N-1. $$
\vskip-0.06in

Take endpoints $n_1, n_2$ of  an edge of the graph $\tilde {\mathcal G}({\bf x}, {\bf W}, L)$,
and select $1\le r\le R$ and $0\le m\le N/L-1$ so that
\eqref{n1n2.eq} holds.  
Set $c:=\sum_{k=0}^{N-1}|X_{w_r}(Lm,k)|^2 e^{i2\pi k (l({\bf w}_r)-1)/N}$
and $c_{\pmb \epsilon}:= \sum_{k=0}^{N-1}(|X_{w_r}(Lm,k)|^2 +{\pmb \epsilon}(r, m, k)) e^{i2\pi k (l({\bf w}_r)-1)/N}$.
Then
\begin{equation}  \label{error.thm.pf.eq2}
|c|\ge \frac{\|{\bf W}\|_*}{N} \min_{n\in V({\bf x})} |x(n)|^2
\end{equation}
\vskip-0.06in
\noindent by \eqref{sufficient.mainthm.pf.eq4-}, and
\begin{equation}\label{error.thm.pf.eq1}
|c_{\pmb \epsilon}-c|\le N |{\pmb \epsilon}|.  
\end{equation}
\vskip-0.06in
Combining  \eqref{sufficient.mainthm.pf.eq4}, \eqref{error.thm.pf.eq2} and
\eqref{error.thm.pf.eq1}, we obtain
\begin{eqnarray*}\label{phasedif.error}
 &\hskip-0.08in  &  \hskip-0.08in\Big|\frac{x_{\pmb \epsilon} (n_1)}{|x_{\pmb \epsilon}(n_1)|}\frac{\overline{x_{\pmb \epsilon} (n_2)}}{|x_{\pmb \epsilon}(n_2)|}-\frac{x(n_1)}{|x(n_1)|}\frac{\overline{x(n_2)}}{|x(n_2)|}\Big|\nonumber\\
& \hskip-0.08in = &  \hskip-0.08in \Big| \frac{c_{\pmb \epsilon}}{|c_{\pmb \epsilon}|}-\frac{c}{|c|}\Big|
\le \frac{|c_{\pmb \epsilon}-c|+ ||c_{\pmb \epsilon}|-|c|| }{|c|} \nonumber\\ 
& \hskip-0.08in \le & \hskip-0.08in \frac{2N^2|{\pmb \epsilon}|}
{\|{\bf W}\|_* \min_{n\in V({\bf x})} |(x(n)|^2   } .
\end{eqnarray*}
This, together with  the fact that the diameter of the graph $\tilde {\mathcal G}({\bf x}, {\bf W}, L)$ is at most $N$,
proves \eqref{error.thm.eq3}.
\end{proof}

\begin{remark} {\rm  Let $\tilde {\bf x}_{\pmb \epsilon}=(\tilde x_{\pmb \epsilon} (0), \ldots, \tilde x_{\pmb \epsilon}(N-1))^\top$, where
$$|\tilde x_{\pmb \epsilon}(n)|= \left\{\begin{array}{ll}
0,   & {\rm if}  \ |x_{\pmb \epsilon}(n)|  \le \frac{1}{2}\min_{n\in V({\bf x})} |x(n)| \\
|x_{\pmb \epsilon}(n)|, &  {\rm otherwise}\end{array}\right.$$
for $0\le n\le N-1$.  Then it follows from  \eqref{error.thm.eq1} and \eqref{error.thm.eq2}  that $V({\bf x})=V(\tilde {\bf x}_{\pmb \epsilon})$
and
$\tilde G({\bf x}, {\bf W}, L)=\tilde G(\tilde {\bf x}_{\pmb \epsilon}, {\bf W}, L)$. 
This indicates that the graph $\tilde G({\bf x}, {\bf W}, L)$ associated with the original signal ${\bf x}$ 
could be found in the noisy environment.
}
\end{remark}



\section {Conclusions}

For multiple windows with small supporting lengths, 
certain constraints on
the support   of a signal could be crucial for its phase retrievability from magnitude measurements of its multiple-window STFT.
The proposed reconstruction algorithm  from corrupted magnitude measurements  yields a good approximation to the original signal
 if we have some priori knowledge on the  noise level and the  minimal magnitude of nonzero components of the original signal.

\end{document}